\RequirePackage{amsmath}
\documentclass[envcountsame,runningheads]{llncs}
\usepackage{graphicx}
\usepackage{url}
\usepackage[dvipsnames]{xcolor}
\usepackage{amsmath,amssymb}
\usepackage{comment}
\usepackage[T1]{fontenc}
\usepackage{mathtools}
\usepackage{complexity}
\usepackage{stmaryrd}
\usepackage{hyperref}
\usepackage{booktabs}
\spnewtheorem{clm}[theorem]{Claim}{\bfseries}{\itshape}

\newcommand{\bigoh}[1]{\mathit{O}(#1)}
\newcommand{\bigomega}[1]{\Omega(#1)}
\newcommand{\set}[1]{\{#1\}}
\newcommand{\ib}[1]{\llbracket#1\rrbracket}

\newcommand{\Res}{\mathit{Res}}
\newcommand{\Tseitin}{\mathit{Tseitin}}
\newcommand{\Ladder}{\mathit{Ladder}}

\newcommand{\CNDRD}{C_{\mathit{ND,RD}}}
\newcommand{\CVSPS}{C_{\mathit{VS,PS}}}
\newcommand{\CSS}{C_{\mathit{S,S}}}
\newcommand{\CNDND}{C_{\mathit{ND,ND}}}
\newcommand{\DNDarb}{D_{\mathit{ND,*}}}
\newcommand{\DNDND}{D_{\mathit{ND,ND}}}
\newcommand{\DNDRD}{D_{\mathit{ND,RD}}}

\begin{document}
\title{Towards a Complexity-theoretic Understanding\texorpdfstring{\\ of Restarts in SAT solvers}{of Restarts in SAT solvers}}
\titlerunning{Complexity-theoretic Understanding of Restarts in SAT solvers}
%
\author{Chunxiao Li\inst{1} \and
Noah Fleming\inst{2} \and Marc Vinyals\inst{3} \and\\ Toniann Pitassi\inst{2} \and
Vijay Ganesh\inst{1}}
\authorrunning{C. Li et al.}
%
\institute{University of Waterloo, Canada 
\and 
University of Toronto, Canada
\and
Technion, Israel}
\maketitle              
\begin{abstract}
Restarts are a widely-used class of techniques integral to the efficiency of Conflict-Driven Clause Learning (CDCL) Boolean SAT solvers. While the utility of such policies has been well-established empirically, a theoretical explanation of whether restarts are indeed crucial to the power of CDCL solvers is lacking.

\vspace{0.10cm}
In this paper, we prove a series of theoretical results that characterize the power of restarts for various models of SAT solvers. More precisely, we make the following contributions. First, we prove an exponential separation between a {\it drunk} randomized CDCL solver model with restarts and the same model without restarts using a family of satisfiable instances. Second, we show that the configuration of CDCL solver with VSIDS branching and restarts (with activities erased after restarts) is exponentially more powerful than the same configuration without restarts for a family of unsatisfiable instances. To the best of our knowledge, these are the first separation results involving restarts in the context of SAT solvers. Third, we show that restarts do not add any proof complexity-theoretic power vis-a-vis a number of models of CDCL and DPLL solvers with non-deterministic static variable and value selection.
\end{abstract}
\section{Introduction}

Over the last two decades, Conflict-Driven Clause Learning (CDCL) SAT solvers have had a revolutionary impact on many areas of software engineering, security and AI. This is primarily due to their ability to solve real-world instances containing millions of variables and clauses~\cite{marques1999grasp,moskewicz2001chaff,biere2009handbook,pipatsrisawat2011power,atserias2011clause}, despite the fact that the Boolean SAT problem is known to be an $\NP$-complete problem and is believed to be intractable in the worst case.

This remarkable success has prompted complexity theorists to seek an explanation for the efficacy of CDCL solvers, with the aim of bridging the gap between theory and practice. Fortunately, a few results have already been established that lay the groundwork for a deeper understanding of SAT solvers viewed as proof systems~\cite{BeameKS04,HertelBPG08,BussHJ08}. Among them, the most important result is the one by Pipatsrisawat and Darwiche~\cite{pipatsrisawat2011power} and independently by Atserias et al.~\cite{atserias2011clause}, that shows that an idealized model of CDCL solvers with non-deterministic branching (variable selection and value selection), and restarts is {\it polynomially equivalent} to the general resolution proof system. However, an important question that remains open is whether this result holds even when restarts are disabled, i.e., whether configurations of CDCL solvers without restarts (when modeled as proof systems) are polynomial equivalent to the general resolution proof system. In practice there is significant evidence that restarts are crucial to solver performance.

This question of the ``power of restarts'' has prompted considerable theoretical work. For example, Bonet, Buss and Johannsen~\cite{bonet2014improved} showed that CDCL solvers with no restarts (but with non-deterministic variable and value selection) are strictly more powerful than regular resolution. Despite this progress, the central questions, such as whether restarts are integral to the efficient simulation of general resolution by CDCL solvers, remain open. 

In addition to the aforementioned theoretical work, there have been many empirical attempts at understanding restarts given how important they are to solver performance. Many hypotheses have been proposed aimed at explaining the power of restarts. Examples include, the {\it heavy-tail} explanation~\cite{gomes2000heavy}, and the ``restarts compact assignment trail and hence produce clauses with lower literal block distance (LBD)'' perspective~\cite{liang2018machine}. Having said that, the heavy-tailed distribution explanation of the power of restarts is not considered valid anymore in the CDCL setting~\cite{liang2018machine}.

 \subsection{Contributions}

In this paper we make several contributions to the theoretical understanding of the power of restarts for several restricted models of CDCL solvers:

\begin{enumerate}
    \item First, we show that CDCL solvers with backtracking, non-deterministic dynamic variable selection, randomized value selection, and restarts\footnote{In keeping with the terminology from \cite{alekhnovich2006exponential}, we refer any CDCL solver with randomized value selection as a \emph{drunk} solver.} are exponentially faster than the same model, but without restarts, with high probability (w.h.p)\footnote{We say that an event occurs with high probability (w.h.p.) if the probability of that event happening goes to $1$ as $n \rightarrow \infty$.}. A notable feature of our proof is that we obtain this separation on a family of satisfiable instances. (See Section~\ref{section: CDCL-ND} for details.)
    
    \item Second, we prove that CDCL solvers with VSIDS variable selection, phase saving value selection and restarts (where activities of variables are reset to zero after restarts) are exponentially faster (w.h.p) than the same solver configuration but without restarts for a class of unsatisfiable formulas. This result holds irrespective of whether the solver uses backtracking or backjumping. (See Section~\ref{section:CDCL-VSIDS} for details.)
    
    \item Finally, we prove several smaller separation and equivalence results for various configurations of CDCL and DPLL solvers with and without restarts. For example, we show that CDCL solvers with non-deterministic static variable selection, non-deterministic static value selection, and with restarts, are polynomially equivalent to the same model but without restarts. Another result we show is that for DPLL solvers, restarts do not add proof theoretic power as long as the solver configuration has non-deterministic dynamic variable selection. (See Section~\ref{section: assorted_results} for details.)
\end{enumerate}

\section{Definitions and Preliminaries}
Below we provide relevant definitions and concepts used in this paper. We refer the reader to the Handbook of Satisfiability~\cite{biere2009handbook} for literature on CDCL and DPLL solvers and to~\cite{Krajicek19ProofComplexity,BeameP01} for literature on proof complexity. 

We denote by $[c]$ the set of natural numbers $\{1,\ldots, c\}$. We treat CDCL solvers as proof systems. For proof systems $A$ and $B$, we use $A \sim_p B$ to denote that they are polynomially equivalent ($p$-equivalent). Throughout this paper it is convenient to think of the trail $\pi$ of the solver during its run on a formula $F$ as a restriction to that formula. We call a function $\pi\colon\{x_1,\ldots,x_n\} \rightarrow \{0,1,*\}$ a \textit{restriction}, where $*$ denotes that the variable is unassigned by $\pi$. Additionally, we assume that our Boolean Constraint Propagation (BCP) scheme is greedy, i.e., BCP is performed till ``saturation''.

\noindent{\bf Restarts in SAT solvers.} A restart policy is a method that erases  part of the state of the solver at certain intervals during the run of a solver~\cite{gomes2000heavy}. In most modern CDCL solvers, the restart policy erases the assignment trail upon invocation, but may choose not to erase the learnt clause database or variable activities. Throughout this paper, we assume that all restart policies are non-deterministic, i.e., the solver may (dynamically) non-deterministically choose its restart sequence. We refer the reader to a paper by Liang et al.~\cite{liang2018machine} for a detailed discussion on modern restart policies.
\section{Notation for Solver Configurations Considered} \label{subsection: notation for solvers}

In this section, we precisely define the various heuristics used to define SAT solver configurations in this paper. By the term {\it solver configuration} we mean a solver parameterized with appropriate heuristic choices. For example, a CDCL solver with non-deterministic variable and value selection, as well as asserting learning scheme with restarts would be considered a solver configuration. 

To keep track of these configurations, we denote solver configurations by the notation $M_{A,B}^{E,R}$, where $M$ indicates the underlying \emph{solver model} (we use $C$ for CDCL and $D$ for DPLL solvers); the subscript $A$ denotes a \emph{variable selection scheme}; the subscript $B$ is a \emph{value selection scheme}; the superscript $E$ is a \emph{backtracking scheme}, and finally the superscript $R$ indicates whether the solver configuration comes equipped with a restart policy. That is, the presence of the superscript $R$ indicates that the configuration has restarts, and its absence indicates that it does not. A $*$ in place of $A,B$ or $E$ denotes that the scheme is \emph{arbitrary}, meaning that it works for any such scheme. See Table~\ref{table: solver configurations} for examples of solver configurations studied in this paper.

\begin{table}[t]
\centering
\caption{Solver configurations in the order they appear in the paper. ND stands for non-deterministic dynamic.}
\label{table: solver configurations}
\renewcommand*{\arraystretch}{1.2}
\begingroup
\setlength{\tabcolsep}{4pt} 
\begin{tabular}{|l|l|l|l|l|l|}
\toprule
& Model & Variable Selection        & Value Selection & Backtracking & Restarts \\
\midrule
$\CNDRD^{T, R}$ & CDCL         & ND & Random Dynamic            & Backtracking        & Yes      \\
$\CNDRD^T$ & CDCL         & ND & Random Dynamic            & Backtracking        & No       \\
$\CVSPS^{J, R}$ & CDCL         & VSIDS                     & Phase Saving              & Backjumping         & Yes      \\
$\CVSPS^J$      & CDCL         & VSIDS                     & Phase Saving              & Backjumping         & No       \\
$\CSS^{J,R}$    & CDCL         & Static                    & Static                    & Backjumping         & Yes      \\
$\CSS^J$        & CDCL         & Static                    & Static                    & Backjumping         & No       \\
$\DNDarb^T$     & DPLL         & ND & Arbitrary                 & Backtracking        & No       \\
$\DNDND^{T, R}$ & DPLL         & ND & ND & Backtracking        & Yes      \\
$\DNDND^T$      & DPLL         & ND & ND & Backtracking        & No       \\
$\DNDRD^{T, R}$ & DPLL         & ND & Random Dynamic            & Backtracking        & Yes      \\
$\DNDRD^T$      & DPLL         & ND & Random Dynamic            & Backtracking        & No       \\
$\CNDND^{J, R}$ & CDCL         & ND & ND & Backjumping         & Yes      \\
$\CNDND^J$      & CDCL         & ND & ND & Backjumping         & No  \\
\bottomrule
\end{tabular}
\endgroup
\end{table}

\subsection{Variable Selection Schemes}
    \noindent{\bf 1. Static (S):} Upon invocation, the S variable selection heuristic returns the unassigned variable with the highest rank according to some predetermined, fixed, total ordering of the variables.
    
    \noindent{\bf 2. Non-deterministic Dynamic (ND):} The ND variable selection scheme non-deterministically selects and returns an unassigned variable.
    
    \noindent{\bf 3. VSIDS (VS)~\cite{moskewicz2001chaff}:} Each variable has an associated number, called its {\it activity}, initially set to 0. Each time a solver learns a conflict, the activities of variables appearing on the conflict side of the implication graph receive a constant bump. The activities of all variables are decayed by a constant c, where $0 < c < 1$, at regular intervals. The VSIDS variable selection heuristic returns the unassigned variable with highest activity, with ties broken randomly.

\subsection{Value Selection Schemes}
    \noindent{\bf 1. Static (S):} Before execution, a 1-1 mapping of variables to values is fixed. The S value selection heuristic takes as input a variable and returns the value assigned to that variable according to the predetermined mapping.
    
    \noindent{\bf 2. Non-deterministic Dynamic (ND):} The ND value selection scheme non-deterministically selects and returns a truth assignment.
    
    \noindent{\bf 3. Random Dynamic (RD):} A randomized algorithm that takes as input a variable and returns a uniformly random truth assignment.
    
    \noindent{\bf 4. Phase Saving (PS):} A heuristic that takes as input an unassigned variable and returns the previous truth value that was assigned to the variable. Typically solver designers determine what value is returned when a variable has not been previously assigned. For simplicity, we use the phase saving heuristic that returns 0 if the variable has not been previously assigned.

\subsection{Backtracking and Backjumping Schemes}
To define different backtracking schemes we use the concept of \emph{decision level} of a variable $x$, which is the number of decision variables on the trail prior to $x$.
	\noindent{\bf Backtracking (T):} Upon deriving a conflict clause, the solver undoes the most recent decision variable on the assignment trail.
	\noindent{\bf Backjumping (J):} Upon deriving a conflict clause, the solver undoes all decision variables with decision level higher than the variable with the second highest decision level in the conflict clause.

\noindent{\bf Note on Solver Heuristics.} Most of our results hold irrespective of the choice of deterministic asserting clause learning schemes (except for Proposition~\ref{clm:WDLS}). Additionally, it goes without saying that the questions we address in this paper make sense only when it is assumed that solver heuristics are polynomial time methods.

\section{Separation for Drunk CDCL with and without Restarts} \label{section: CDCL-ND}

Inspired by Alekhnovich et al.~\cite{alekhnovich2006exponential}, where the authors proved exponential lower bound for drunk DPLL solvers over a class of satisfiable instances, we studied the behavior of restarts in a drunken model of CDCL solver. We introduce a class of satisfiable formulas, $\Ladder_n$, and use them  to prove the separation between $\CNDRD^{T,R}$ and $\CNDRD^{T}$.
At the core of these formulas is a formula which is hard for general resolution even after any small restriction (corresponding to the current trail of the solver). For this, we use the well-known Tseitin formulas. 

\begin{definition}[Tseitin Formulas]
    Let $G=(V,E)$ be a graph and $f\colon V \rightarrow \{0,1\}$ a labelling of the vertices. The formula $\Tseitin(G,f)$ has variables $x_e$ for $e \in E$ and constraints $\bigoplus_{uv \in E} x_{uv} = f(v)$ for each $v \in V$.
\end{definition}
For any graph $G$, $\Tseitin(G,f)$ is unsatisfiable iff $\bigoplus_{v \in V} f(v) = 1$, in which case we call $f$ an \emph{odd labelling}. The specifics of the labelling are irrelevant for our applications, any odd labelling will do. Therefore, we often omit defining $f$, and simply assume that it is odd.

The family of satisfiable $\Ladder_n$ formulas are built around the Tseitin formulas, unless the variables of the formula are set to be consistent to one of two satisfying assignments, the formula will become unsatisfiable. Furthermore, the solver will only be able to backtrack out of the unsatisfiable sub-formula by first refuting Tseitin, which is a provably hard task for any CDCL solver~\cite{Urquhart87}.

The $\Ladder_n$ formulas contain two sets of variables, $\ell^i_j$ for $0 \leq i \leq n-2, j \in [\log n]$ and $c_m$ for $m \in [\log n]$, where $n$ is a power of two. We denote by $\ell^i$ the block of variables $\{\ell^i_1,\ldots, \ell^i_{\log n}\}$.
These formulas are constructed using the following gadgets.

\noindent {\bf Ladder gadgets}: $L^i := (\ell^i_1 \vee \ldots \vee \ell^i_{\log n}) \wedge (\neg \ell^i_1 \vee \ldots \vee \neg \ell^i_{\log n})$.\\ Observe that $L^i$ is falsified only by the all-$1$ and all-$0$ assignments.

\noindent {\bf Connecting gadgets}: $C^i := (c_1^{\mathit{bin}(i,1)} \wedge \ldots \wedge c_{\log n}^{\mathit{bin}(i,\log n)})$. \\ 
 Here, $\mathit{bin}(i,m)$ returns the $m$th bit of the binary representation of $i$, and $c_m^1 := c_m$, while $c_m^0 := \neg c_m$. That is, $C^i$ is the conjunction that is satisfied only by the assignment encoding $i$ in binary.
 
\noindent {\bf Equivalence gadget}: $EQ:= \bigwedge_{i,j =0}^{n-2}\bigwedge_{m,k=1}^{\log n} (\ell^i_k \iff \ell^j_m)$. \\ These clauses 
enforce that every $\ell$-variable must take the same value. 



\begin{definition}[Ladder formulas] For $G=(V,E)$ with $|E|=n-1$ where $n$ is a power of two, let $\Tseitin(G,f)$ be defined on the variables $\{\ell^0_1,\ldots, \ell^{n-2}_1\}$.
 $\Ladder_n(G,f)$ is the conjunction of the clauses representing
\begin{align*}
    &L^i \Rightarrow C^i, &\forall 0 \leq i \leq n-2 \\
    &C^i \Rightarrow \Tseitin(G,f), &\forall 0 \leq i \leq n-2 \\
    &C^{n-1} \Rightarrow EQ.
\end{align*}
\end{definition}
Observe that the $\Ladder_n(G,f)$ formulas have polynomial size provided that the  degree of $G$ is $O(\log n)$. As well, this formula is satisfiable only by the assignments that sets $c_m = 1$ and $\ell^i_j = \ell^p_q$ for every $m,j,q \in [\log n]$ and  $0 \leq i,p \leq n-2$.

These formulas are constructed so that after setting only a few variables,
any drunk solver will enter an unsatisfiable subformula w.h.p. and thus be forced to refute the Tseitin formula.  Both the
ladder gadgets and equivalence gadget act as trapdoors for the Tseitin formula. Indeed, if any $c$-variable is set to $0$ then we have already entered an unsatisfiable instance. Similarly, setting $\ell^i_j=1$ and $\ell^p_q =0$ for any $0 \leq i,p \leq n-2$, $j,q \in [\log n]$ causes us to enter an unsatisfiable instance. This is because setting all $c$-variables to $1$ together with this assignment would falsify a clause of the equivalence gadget. Thus, after the second decision of the solver, the probability that it is in an unsatisfiable instance is already at least $1/2$. With these formulas in hand, we prove the following theorem, separating backtracking $\CNDRD^T$ solvers with and without restarts.

\begin{theorem}\label{thm:main_ND_RD}
    There exists a family of $O(\log n)$-degree graphs $G$ such that 
    \begin{enumerate} 
    	\item $\Ladder_n(G,f)$ can be decided in time $O(n^2)$ by $\CNDRD^{T,R}$, except with exponentially small probability.
    	\item $\CNDRD^T$ requires exponential time to decide $\Ladder_n(G,f)$, except with probability $O(1/n)$.
    \end{enumerate}
\end{theorem}
The proof of the preceding theorem occupies the remainder of this section.

\subsection{Upper Bound on Ladder Formulas Via Restarts.}
We present the proof for part (1) of Theorem~\ref{thm:main_ND_RD}. The proof relies on the following lemma, stating that given the all-$1$ restriction to the $c$-variables, $\CNDRD^T$ will find a satisfying assignment.

\begin{lemma} \label{lemma: best restriction}
For any graph $G$, $\CNDRD^T$ will find a satisfying assignment to \\ $\Ladder_n(G,f)[c_1= 1, \ldots, c_{\log n} =1 ]$ in time $O(n \log n)$.
\end{lemma}

\begin{proof}
When all $c$ variables are $1$, we have $C^{n-1} = 1$. By the construction of the connecting gadget, $C^{i} = 0$ for all $0 \leq i \leq n-2$. Under this assignment, the remaining clauses belong to $EQ$, along with $\neg L^i$ for $0 \le i \le n-2$. It is easy to see that, as soon as the solver sets an $\ell$-variable, these clauses will propagate the remaining  $\ell$-variables to the same value.
\qed
\end{proof}

Put differently, the set of $c$ variables forms a \emph{weak backdoor}~\cite{williams2003connections,williams2003backdoors} for $\Ladder_n$ formulas. Part (1) of Theorem~\ref{thm:main_ND_RD}  shows that, with probability at least $1/2$, $\CNDRD^{T,R}$ can exploit this weak backdoor using only $O(n)$ number of restarts.

\begin{proof}[of Theorem~\ref{thm:main_ND_RD} Part (1)]
By Lemma~\ref{lemma: best restriction}, if $\CNDRD^{T,R}$ is able to assign all $c$ variables to $1$ before assigning any other variables, then the solver will find a satisfying assignment in time $O(n \log n)$ with probability 1. We show that the solver can exploit restarts in order to find this assignment. The strategy the solver adopts is as follows: query each of the $c$-variables; if at least one of the $c$-variables was assigned to 0, restart. We argue that if the solver repeats this procedure $k = n^2$ times then it will find the all-$1$ assignment to the $c$-variables, except with exponentially small probability. Because each variable is assigned $0$ and $1$ with equal probability, the probability that a single round of this procedure finds the all-$1$ assignment is $2^{-\log n}$. Therefore, the probability that the solver has not found the all-$1$ assignment after $k$ rounds is
\[ (1-1/n)^k \leq e^{-k/n} = e^{-n}. \eqno\qed\]
\end{proof}


\subsection{Lower Bound on Ladder Formulas Without Restarts}
We now prove part (2) of Theorem~\ref{thm:main_ND_RD}.
The proof relies on the following three technical lemmas. The first claims that the solver is well-behaved (most importantly that it cannot learn any new clauses) while it has not made many decisions.

\begin{lemma}\label{lemma: solver_cant_screw_us}
	Let $G$ be any graph of degree at least $d$. Suppose that $\CNDRD^T$ has made $\delta < \min (d-1, \log n-1)$ decisions since its invocation on $\Ladder_n(G,f)$. Let $\pi_\delta$ be the current trail, then
	\begin{enumerate}
		\item The solver has yet to enter a conflict, and thus has not learned any clauses.
		\item The trail $\pi_\delta$ contains variables from at most $\delta$ different blocks $\ell^i$.
	\end{enumerate}
\end{lemma}
The proof of this lemma is deferred to the appendix.

The following technical lemma states that if a solver with backtracking has caused the formula to become unsatisfiable, then it must \emph{refute} that formula before it can backtrack out of it.
For a restriction $\pi$ and a formula $F$, we say that the solver has \emph{produced a refutation} of an unsatisfiable formula $F[\pi]$ if it has learned a clause $C$ such that $C$ is falsified under $\pi$. Note that because general resolution $p$-simulates CDCL, any refutation of a formula $F[\pi]$ implies a general resolution refutation of $F[\pi]$ of size at most polynomial in the time that the solver took to produce that refutation.

\begin{lemma}~\label{lemma: backtracking_must_refute}
	 Let $F$ be any propositional formula, let $\pi$ be the current trail of the solver, and let $x$ be any literal in $\pi$. Then, $\CNDND^{T}$ backtracks $x$ only after it has produced a refutation of $F[\pi]$.
\end{lemma}
\begin{proof}
    In order to backtrack $x$, the solver must have learned a clause $C$ asserting the negation of some literal $z \in \pi$ that was set before $x$. Therefore, $C$ must only contain the negation of literals in $\pi$. Hence, $C[\pi] = \emptyset$. \qed 
\end{proof}

The third lemma reduces proving a lower bound on the runtime of $\CNDND^T$ on the $\Ladder_n$ formulas under any well-behaved restriction to proving a general resolution lower bound on an associated Tseitin formula.
\begin{definition}
    For any unsatisfiable formula $F$, denote by $\Res(F \vdash \emptyset)$ the minimal size of any general resolution refutation of $F$.
\end{definition}

We say that a restriction (thought of as the current trail of the solver) $\pi$ to $\Ladder_n(G,f)$ \emph{implies Tseitin} if $\pi$ either sets some $c$-variable to $0$ or $\pi[\ell^i_j]=1$ and $\pi[\ell^p_q] =0$ for some $0\leq i,q \leq n-2$, $j,q \in [\log n]$. Observe that in both of these cases the formula $\Ladder_n(G,f)[\pi]$ is unsatisfiable. 

	\begin{lemma} 
	\label{lemma:reduction_to_tseitin}
	Let $\pi$ be any restriction that implies Tseitin and such that each clause of $\Ladder_n(G,f)[\pi]$ is either satisfied or contains at least two unassigned variables.
	Suppose that $\pi$ sets variables from at most $\delta$ blocks $\ell^i$. Then there is a restriction $\rho_\pi^*$  that sets at most $\delta$ variables of $\Tseitin(G,f)$ such that
	\[ \Res(\Ladder_n(G,f)[\pi] \vdash \emptyset ) \geq \Res(\Tseitin(G,f)[\rho_\pi^*] \vdash \emptyset ).\]
\end{lemma}

We defer the proof of this lemma to the appendix, and show how to use them to prove part (2) of Theorem~\ref{thm:main_ND_RD}. We prove this statement for any degree $O(\log n)$ graph $G$ with sufficient expansion. 
\begin{definition}
    The expansion of a graph $G = (V,E)$ is  
\[ e(G) := \min_{V' \subseteq V, |V'| \leq |V|/2} \frac{|E[V', V \setminus V']|}{ |V'|}, \]
where $E[V', V \setminus V']$ is the set of edges in $E$ with one endpoint in $V'$ and the other in $V \setminus V'$.
\end{definition}
For every $d\geq 3$, \emph{Ramanujan Graphs} provide an infinite family of $d$-regular expander graphs $G$ for which $e(G) \geq d/4$. The lower bound on solver runtime relies on the general resolution lower bounds for the Tseitin formulas~\cite{Urquhart87}; we use the following lower bound criterion which follows immediately\footnote{In particular, this follows from Theorem 4.4 and Corollary 3.6 in \cite{ben2001short}, noting that the definition of expansion used in their paper is lower bounded by $3e(G)/|V|$ as they restrict to sets of vertices of size between $|V|/3$ and $2|V|/3$.} from \cite{ben2001short}.

\begin{corollary}[\cite{ben2001short}]~\label{corollary: BW}
	For any connected graph $G=(V,E)$ with maximum degree $d$ and odd weight function $f$, \[\Res(\Tseitin(G, f) \vdash \emptyset) = \exp \left(\Omega \bigg( \frac{(e(G) |V|/3 - d)^2}{|E|} \bigg) \right)\]
\end{corollary}
We are now ready to prove the theorem. 

\begin{proof}[of part (2) Theorem~\ref{thm:main_ND_RD}]
	Fix $G=(V,E)$ to be any degree-$(8 \log n)$ graph on $|E| = n-1$ edges such that $e(G) \geq 2 \log n$. Ramanujan graphs satisfy these conditions.
	
	First, we argue that  within $\delta < \log n-1$ decisions from the solver's invocation, the trail $\pi_\delta$ will imply Tseitin,
	except with probability $1-/2^{\delta-1}$. By Lemma~\ref{lemma: solver_cant_screw_us}, the solver has yet to backtrack or learn any clauses, and it has set variables from at most $\delta$ blocks $\ell^i$. Let $x$ be the variable queried during the $\delta$th decision. If $x$ is a $c$ variable, then with probability $1/2$ the solver sets $c_i = 0$. If $x$ is a variable $\ell^i_j$, then, unless this is the first time the solver sets an $\ell$-variable, the probability that it sets $\ell^i_j$ to a different value than the previously set $\ell$-variable is $1/2$. 
	
	Conditioning on the event that, within the first $\log n-2$ decisions the trail of the solver implies Tseitin
	(which occurs with probability at least $(n-8)/n$), we argue that the runtime of the solver is exponential in $n$. Let $\delta < \log n-1$ be the first decision level such that the current trail $\pi_\delta$ implies Tseitin. By Lemma~\ref{lemma: backtracking_must_refute} the solver must have produced a refutation of $\Ladder_n(G,f)[\pi_\delta]$ in order to backtrack out of the unsatisfying assignment. If the solver takes $t$ steps to refute $\Ladder_n(G,f)[\pi_\delta]$ then this implies a general resolution refutation of size $\poly(t)$. Therefore, in order to lower bound the runtime of the solver, it is enough to lower bound the size of general resolution refutations of $\Ladder_n(G,f)[\pi_\delta]$.
	
	 By Lemma~\ref{lemma: solver_cant_screw_us}, the solver has not learned any clauses, and has yet to enter into a conflict and therefore no clause in $\Ladder_n(G,f)[\pi_\delta]$ is falsified. As well, $\pi_\delta$ sets variables from at most $\delta < \log n-1$ blocks $\ell^i$. By Lemma~\ref{lemma:reduction_to_tseitin} there exists a restriction $\rho_\pi^*$ such that $\Res(\Ladder_n(G,f)[\pi] \vdash \emptyset) \geq \Res(\Tseitin(G,f)[\rho_\pi^*] \vdash \emptyset)$. Furthermore, $\rho_\pi^*$ sets at most $\delta < \log n-1$ variables and therefore cannot falsify any constraint of $\Tseitin(G,f)$, as each clause depends on $8 \log n$ variables.
	 Observe that if we set a variable $x_e$ of $\Tseitin(G,f)$ then we obtain a new instance of $\Tseitin(G_{\rho^*_\pi},f')$ on a graph $G_{\rho^*_\pi}=(V, E \setminus \{e\})$. Therefore, we are able to apply Corollary~\ref{corollary: BW} provided that we can show that $e(G_{\rho^*_\pi})$ is large enough.
	 
	 
	 \begin{clm}
	 	Let $G = (V,E)$ be a graph and let $G'=(V,E')$ be obtained from $G$ by removing at most $e(G)/2$ edges. Then $e(G') \geq e(G)/2$. 	
	 \end{clm}
	 \begin{proof}
	 	Let $V' \subseteq V$ with $|V'| \leq |V|/2$. Then, $E'[V', V \setminus V'] \geq e(G)|V'| - e(G)/2 \geq (e(G)/2)|V'|$. \qed
	 \end{proof}
	 
	 It follows that $e(G_{\rho^*_\pi}) \geq \log n$. Note that $|V| = n/8\log n$. By Corollary~\ref{corollary: BW}, 
	 \[ \Res(\Ladder_n(G,f)[\pi] \vdash \emptyset) = \exp ( \Omega( ((n-1)/24 - 8\log n)^2/n )) = \exp(\Omega(n)). \]
	 Therefore, the runtime of $\CNDND^{T}$ is $\exp(\Omega(n))$ on $Ladder_n(G,F)$ w.h.p. \qed


\end{proof}



\section{CDCL+VSIDS Solvers with and without Restarts}
\label{section:CDCL-VSIDS}

In this section, we prove that CDCL solvers with VSIDS variable selection, phase saving value selection and restarts (where activities of variables are reset to zero after restarts) are exponentially more powerful than the same solver configuration but without restarts, w.h.p. 

\begin{theorem}
  \label{lem:vsids-separation}
  There is a family of unsatisfiable formulas that can be decided in polynomial
  time with $\CVSPS^{J, R}$ but requires exponential time with
  $\CVSPS^J$, except with exponentially small probability.
\end{theorem}

We show this separation using pitfall formulas $\Phi(G_n,f,n,k)$,
designed to be hard for solvers using
the VSIDS heuristic~\cite{vinyals20hard}. We assume that $G_n$ is a constant-degree
expander graph with $n$ vertices and $m$ edges,
$f\colon V(G_n)\to\set{0,1}$ is a function with odd support as with
Tseitin formulas, we think of $k$ as a constant and let $n$
grow. We denote the indicator function of a Boolean expression $B$ with $\ib{B}$. These formulas have $k$ blocks of variables named $X_j$, $Y_j$,
$Z_j$, $P_j$, and $A_j$, with $j\in[k]$, and the following clauses:
\begin{itemize}
\item $\left(\bigoplus_{e \ni v} x_{j,e} = f(v)\right) \lor \bigvee_{i=1}^n z_{j,i}$, expanded into CNF, for $v\in V(G_n)$ and $j\in[k]$;
\item $y_{j,i_1} \lor y_{j,i_2} \lor \neg p_{j,i_3}$ for $i_1,i_2\in[n]$, $i_1 < i_2$, $i_3 \in [m+n]$, and $j\in [k]$;
\item $y_{j,i_1} \lor \bigvee_{i\in [m+n]\setminus \set{i_2}} p_{j,i} \lor \bigvee_{i=1}^{i_2-1} x_{j,i} \lor \neg x_{j,i_2}$ for $i_1 \in [n]$, $i_2 \in [m]$, and $j \in [k]$;
\item $y_{j,i_1} \lor \bigvee_{i\in [m+n]\setminus \set{m+i_2}} p_{j,i} \lor \bigvee_{i=1}^{m} x_{j,i} \lor \bigvee_{i=1+\ib{i_2=n}}^{i_2-1} z_{j,i} \lor \neg z_{j,i_2}$ for $i_1,i_2 \in [n]$ and $j \in [k]$;
\item $\neg a_{j,1} \lor a_{j,3} \lor \neg z_{j,i_1}$, $\neg a_{j,2} \lor \neg a_{j,3} \lor \neg z_{j,i_1}$, $a_{j,1} \lor \neg z_{j,i_1} \lor \neg y_{j,i_2}$, and $a_{j,2} \lor \neg z_{j,i_1} \lor \neg y_{j,i_2}$ for $i_1,i_2 \in [n]$ and $j \in [k]$; and
\item $\bigvee_{j \in [k]} \neg y_{j,i} \lor \neg y_{j,i+1}$ for odd $i\in[n]$.
\end{itemize}

To give a brief overview, the first type of clauses are essentially a Tseitin formula and thus are hard to solve. The next four types form a pitfall gadget, which has the following easy-to-check property.

\begin{clm}
  \label{prop:pair-y-conflict}
Given any pair of variables $y_{j,i_1}$ and $y_{j,i_2}$ from the same block $Y_j$, assigning $y_{j,i_1}=0$ and $y_{j,i_2}=0$ yields a conflict.
\end{clm}

Furthermore, such a conflict involves all of the variables of a block
$X_j$, which makes the solver prioritize these variables and it
becomes stuck in a part of the search space where it must refute the
first kind of clauses. Proving this formally requires a delicate
argument, but we can use the end result as a black box.

\begin{theorem}[{\cite[Theorem~3.6]{vinyals20hard}}]
  \label{lem:vsids-lower-bound}
  For $k$ fixed, $\Phi(G_n,f,n,k)$ requires time $\exp(\bigomega{n})$ to
  decide with $\CVSPS^J$, except with exponentially small
  probability.
\end{theorem}

The last type of clauses, denoted by $\Gamma_i$, ensure that a short general
resolution proof exists. Not only that, we can also prove that pitfall
formulas have small backdoors~\cite{williams2003connections,williams2003backdoors}, which is
enough for a formula to be easy for $\CVSPS^{J,R}$.

\begin{definition}
  \label{def:conflicting-core}
  A set of variables $V$ is a \textit{strong backdoor} for unit-propagation if
  every assignment to all variables in $V$ leads to a conflict,
  after unit propagation.
\end{definition}

\begin{lemma}
  \label{lem:vsids-upper-bound}
  If $F$ has a strong backdoor for unit-propagation of size $c$, then
  $\CVSPS^{J,R}$ can solve $F$ in time $n^{\bigoh{c}}$, except
  with exponentially small probability.
\end{lemma}

\begin{proof}
  We say that the solver learns a beneficial clause if it only
  contains variables in $V$. Since there are $2^c$ possible
  assignments to variables in $V$ and each beneficial clause forbids
  at least one assignment, it follows that learning $2^c$ beneficial
  clauses is enough to produce a conflict at level $0$.

  Therefore it is enough to prove that, after each restart, we learn a
  beneficial clause with large enough probability. Since all variables
  are tied, all decisions before the first conflict after a restart are random, and hence with
  probability at least $n^{-c}$ the first variables to be decided
  before reaching the first conflict are (a subset of) $V$. If this is
  the case then, since $V$ is a strong backdoor, no more decisions are
  needed to reach a conflict, and furthermore all decisions in the
  trail are variables in $V$, hence the learned clause is beneficial.

  It follows that the probability of having a sequence of $n^{2c}$ restarts without learning a beneficial clause is at most
  \begin{equation}
    (1-n^{-c})^{n^{2c}} \leq \exp(-n^{-c}\cdot n^{2c}) = \exp(-n^c)
  \end{equation}
  hence by a union bound the probability of the algorithm needing more than $2^c\cdot n^{2c}$ restarts is at most $2^c \cdot \exp(-n^c)$.
  \qed
\end{proof}

We prove Theorem~\ref{lem:vsids-separation} by showing that $\Phi(G_n,f,n,k)$
contains a backdoor of size $2k(k+1)$.

\begin{proof}[of Theorem~\ref{lem:vsids-separation}]
  We claim that the set of variables
  $V = \set{y_{j,i} \mathrel{|} (j,i) \in [k]\times[2k+2]}$ is a strong backdoor
  for unit-propagation. Consider any assignment to $V$. Each of the
  $k+1$ clauses $\Gamma_1,\Gamma_3,\ldots,\Gamma_{2k+1}$ forces a
  different variable $y_{j,i}$ to $0$, hence by the pigeonhole
  principle there is at least one block with two variables assigned to
  $0$. But by Claim~\ref{prop:pair-y-conflict}, this is enough to
  reach a conflict.

  The upper bound follows from Lemma~\ref{lem:vsids-upper-bound},
  while the lower bound follows from Theorem~\ref{lem:vsids-lower-bound}.
  \qed
\end{proof}


\section{Minor Equivalences and Separations for CDCL/DPLL Solvers with and without Restarts} \label{section: assorted_results}

In this section, we prove four smaller separation and equivalence results for various configurations of CDCL and DPLL solvers with and without restarts.

\subsection{Equivalence between CDCL Solvers with Static Configurations with and without Restarts}
First, we show that CDCL solvers with non-deterministic static variable and value selection without restarts ($\CSS^{J}$) is as powerful as the same configuration with restarts ($\CSS^{J,R}$) for both satisfiable and unsatisfiable formulas. We assume that the BCP subroutine for the solver configurations under consideration is ``fixed'' in the following sense: if there is more than one unit clause under a partial assignment, the BCP subroutine propagates the clause that is added to the clause database first.

\begin{theorem} \label{theorem: NS_NS}
$\CSS^{J} \sim_p \CSS^{J,R}$ provided that they are given the same variable ordering and fixed mapping of variables to values  for the variable selection and value selection schemes respectively.
\end{theorem}
    We prove this theorem by arguing for any run of $\CSS^{J,R}$, that restarts can be removed without increasing the run-time.
\begin{proof}
    Consider a run of $\CSS^{J,R}$ on some formula $F$, and suppose that the solver has made $t$ restart calls. Consider the trail $\pi$ for $\CSS^{J,R}$ up to the variable $l$ from the second highest decision from the last learnt clause before the first restart call. Now, observe that because the decision and variable selection orders are static, once $\CSS^{J,R}$ restarts, it will force it to repeat the same decisions and unit propagations that brought it to the trail $\pi$. Suppose that this is not the case and consider the first literal on which the trails differ. This difference could not be caused by a unit propagation as the solver has not learned any new clauses since the restart. Thus, it  must have been caused by a decision. However, because the clause databases are the same, this would contradict the static variable and value order.
 Therefore, this restart can be ignored, and we obtain a run of $\CSS^{J,R}$ with $d-1$ restarts  without increasing the run-time. The proof follows by induction. Once all restarts have been removed, the result is a valid run of $\CSS^J$.
\qed
\end{proof}

Note that in the proof of Theorem~\ref{theorem: NS_NS}, not only we argue that $\CSS^J$ is $p$-equivalent to $\CSS^{J,R}$, we also show that the two configurations produce the same run. The crucial observation is that given any state of $\CSS^{J,R}$, we can produce a run of $\CSS^J$ which ends in the same state. In other words, our proof not only suggests that $\CSS^{J,R}$ is equivalent to $\CSS^J$ from a proof theoretic point of view, it also implies that the two configurations are equivalent for satisfiable formulas.


\subsection{Equivalence between DPLL Solvers with ND Variable Selection on UNSAT Formulas} 

We show that when considered as a proof system, a DPLL solver with non-deterministic dynamic variable selection, arbitrary value selection and no restarts ($\DNDarb^T$) is p-equivalent to DPLL solver with non-deterministic dynamic variable and value selection and restarts ($\DNDND^{T,R}$), and hence, transitively p-equivalent to tree-like resolution---the restriction of general resolution where each consequent can be an antecedent in only one later inference.

\begin{theorem}\label{theorem: DPLL_ND_ND p-equiv DPLL_ND_arb}
$\DNDarb^T \sim_p \DNDND^T$.
\end{theorem}

\begin{proof}
To show that $\DNDND^T$ $p$-simulates $\DNDarb^T$, we argue that every proof of $\DNDND^T$ can be converted to a proof of same size in $\DNDarb^T$. Let $F$ be an unsatisfiable formula. Recall that a run of $\DNDND^T$ on $F$ begins with non-deterministically picking some variable $x$ to branch on, and a truth value to assign to $x$. W.l.o.g. suppose that the solver assigns $x$ to $1$. Thus, the solver will first refute $F[x = 1]$ before backtracking and refuting $F[x = 0]$.

To simulate a run of $\DNDND^T$ with $\DNDarb^T$, since variable selection is non-deterministic, $\DNDarb^T$ also chooses the variable $x$ as the first variable to branch on. If the value selection returns $x= \alpha$ for $\alpha \in \{0,1\}$, then the solver focus on the restricted formula $F[x = \alpha]$ first. Because there is no clause learning, whether $F[x = 1]$ or $F[x = 0]$ is searched first does not affect the size of the search space for the other. The proof follows by recursively calling $\DNDarb^T$ on $F[x = 1]$ and $F[x = 0]$. The converse direction follows since every run of $\DNDarb^T$ is a run of $\DNDND^T$.
\qed
\end{proof}

\begin{corollary} \label{corollary: DPLL_ND p-sim all}
$\DNDarb^T \sim_p \DNDND^{T,R}$.
\end{corollary}

\begin{proof}
This follows from the fact that $\DNDND^{T,R} \sim_p \DNDND^T$. Indeed, with non-deterministic branching and without clause learning, restarts cannot help. If ever $\DNDND^{T,R}$ queries a variable $x = \alpha$ for $\alpha \in \{0,1\}$ and then later restarts to assign it to $1-\alpha$, then $\DNDND^T$ ignores the part of the computation when $x = \alpha$ and instead immediately non-deterministically chooses $x = 1-\alpha$.
\qed
\end{proof}

It is interesting to note that while the above result establishes a $p$-equivalence between DPLL solver models $\DNDarb^T$ and $\DNDND^{T,R}$, the following corollary implies that DPLL solvers with non-deterministic variable and randomized value selection are exponentially separable for satisfiable instances.

\subsection{Separation Result for Drunk DPLL Solvers}
 We show that DPLL solvers with non-deterministic variable selection, randomized value selection and no restarts ($\DNDRD^T$) is exponentially weaker than the same configuration with restarts ($\DNDRD^{T,R}$).
 
 \begin{corollary}
 $\DNDRD^T$ runs exponentially slower on the class of satisfiable formulas $\Ladder_n(G, f)$ than $\DNDRD^{T,R}$, with high probability.
 \end{corollary}
 The separation follows from the fact that our proof of the upper bound from Theorem~\ref{thm:main_ND_RD} does not use the fact the solver has access to clause learning, which means the solver $\DNDRD^{T,R}$ can also find a satisfying assignment for $\Ladder_n(G,f)$ in time $O(n^2)$, except with exponentially small probability. On the other hand, the lower bound from Theorem~\ref{thm:main_ND_RD} immediately implies an exponential lower bound for $\DNDRD^T$, since $\DNDRD^T$ is strictly weaker than $\CNDRD^T$.

\subsection{Separation Result for CDCL Solvers with WDLS}

Finally, we state an observation of Robert Robere~\cite{Robere19Personal} on restarts in the context of the Weak Decision Learning Scheme (WDLS).

\begin{definition}[WDLS]
Upon deriving a conflict, a CDCL solver with WDLS learns a conflict clause which is the disjunction of the negation of the decision variables on the current assignment trail.
\end{definition}

\begin{theorem}\label{clm:WDLS}
$\CNDND^J$+WDLS is exponentially weaker than $\CNDND^{J,R}$+WDLS.
\end{theorem}
\begin{proof}
The solver model $\CNDND^J$ with WDLS is only as powerful as $\DNDND^T$, since each learnt clause will only be used once for propagation after the solver backtracks immediately after learning the conlict clause, and remains satisfied for the rest of the solver run. This is exactly how $\DNDND^T$ behaves under the same circumstances. On the other hand,  WDLS is an asserting learning scheme~\cite{pipatsrisawat2008new}, and hence satisfies the conditions of the main theorem in~\cite{pipatsrisawat2011power}, proving that CDCL with any asserting learning scheme and restarts p-simulates general resolution. Thus, we immediately have $\CNDND^{J,R}$ with WDLS is exponentially more powerful than the same solver but with no restarts (for unsatisfiable instances).
\qed
\end{proof}
\section{Related Work}

\noindent{\bf Previous Work on Theoretical Understanding of Restarts:} Buss et al.~\cite{BussHJ08} and Van Gelder~\cite{VanGelder05PoolResolution} proposed two proof systems, namely regWRTI and pool resolution respectively, with the aim of explaining the power of restarts in CDCL solvers. Buss et al. proved that regWRTI is able to capture exactly the power of CDCL solvers with {\it non-greedy BCP} and without restarts and Van Gelder proved that pool resolution can simulate certain configurations of DPLL solvers with clause learning. As both pool resolution and regWRTI are strictly more powerful than regular resolution, a natural question is whether formulas that exponentially separate regular and general resolution can be used to prove lower bounds for pool resolution and regWRTI, thus transitively proving lower bounds for CDCL solvers without restarts. However, since Bonet et al.~\cite{bonet2014improved} and Buss and Ko\l{}odziejczyk~\cite{BK14SmallStone} proved that all such candidates have short proofs in pool resolution and regWRTI, the question of whether CDCL solvers without restarts are as powerful as general resolution still remains open.

\noindent{\bf Previous Work on Empirical Understanding of Restarts:} The first paper to discuss restarts in the context of DPLL SAT solvers was by Gomes and Selman~\cite{gomes2000heavy}. They proposed an explanation for the power of restarts popularly referred to as ``heavy-tailed explanation of restarts''. Their explanation relies on the observation that the runtime of randomized DPLL SAT solvers on satisfiable instances, when invoked with different random seeds, exhibits a heavy-tailed distribution. This means that the probability of the solver exhibiting a long runtime on a given input and random seed is non-negligible. However, because of the heavy-tailed distribution of solver runtimes, it is likely that the solver may run quickly on the given input for a different random seed. This observation was the motivation for the original proposal of the restart heuristic in DPLL SAT solvers by Gomes and Selman~\cite{gomes2000heavy}.

Unfortunately, the heavy-tailed explanation of the power of restarts does not lift to the context of CDCL SAT solvers. The key reason is that, unlike DPLL solvers, CDCL solvers save solver state (e.g., learnt clauses and variable activities) across restart boundaries. Additionally, the efficacy of restarts has been observed for both deterministic and randomized CDCL solvers, while the heavy-tailed explanation inherently relies on randomness. Hence, newer explanations have been proposed and validated empirically on SAT competition benchmarks. Chief among them is the idea that ``restarts compact the assignment trail during its run and hence produce clauses with lower literal block distance (LBD), a key metric of quality of learnt clauses''~\cite{liang2018machine}.

\noindent{\bf Comparison of Our Separation Results with Heavy-tailed Explanation of Restarts:} A cursory glance at some of our separation results might lead one to believe that they are a complexity-theoretical analogue of the heavy-tailed explanation of the power of restarts, since our separation results are over randomized solver models. We argue this is not the case. First, notice that our main results are for drunk CDCL solvers that save solver state (e.g., learnt clauses) across restart boundaries, unlike the randomized DPLL solvers studied by Gomes et al.~\cite{gomes2000heavy}. Second, we make no assumptions about independence (or lack thereof) of branching decisions across restarts boundaries. In point of fact, the variable selection in the CDCL model we use is non-deterministic. Only the value selection is randomized. More precisely, we have arrived at a separation result without relying on the assumptions made by the heavy-tailed distribution explanation, and interestingly we are able to prove that the ``solver does get stuck in a bad part of the search space by making bad value selections''. Note that in our model the solver is free to go back to ``similar parts of the search space across restart boundaries''. In fact, in our proof for CDCL with restarts, the solver chooses the same variable order across restart boundaries.

\section{Conclusions}
In this paper, we prove a series of results that establish the power of restarts (or lack thereof) for several models of CDCL and DPLL solvers. We first showed that CDCL solvers with backtracking, non-deterministic dynamic variable selection, randomized dynamic value selection, and restarts are exponentially faster than the same model without restarts for a class of satisfiable instances. Second, we showed CDCL solvers with VSIDS variable selection and phase saving without restarts are exponentially weaker than the same solver with restarts, for a family of unsatisfiable formulas. Finally, we proved four additional smaller separation and equivalence results for various configurations of DPLL and CDCL solvers. 

By contrast to previous attempts at a ``theoretical understanding the power of restarts'' that typically assumed that variable and value selection heuristics in solvers are non-deterministic, we chose to study randomized or real-world models of solvers (e.g., VSIDS branching with phase saving value selection). The choices we made enabled us to more effectively isolate the power of restarts in the solver models we considered. This leads us to the belief that the efficacy of restarts becomes apparent only when the solver models considered have weak heuristics (e.g., randomized or real-world deterministic) as opposed to models that assume that all solver heuristics are non-deterministic.

\bibliographystyle{splncs04}
\bibliography{reference}

\newpage
\section*{Appendix}


\subsection*{1.\ \ \ \ Proof of Lemma~\ref{lemma: solver_cant_screw_us} (See page \pageref{lemma: solver_cant_screw_us})}

\textbf{Lemma~\ref{lemma: solver_cant_screw_us}. } \label{appendix_lemma_3}
{\it
	Let $G$ be any graph of degree at least $d$. Suppose that $\CNDRD$ has made $\delta < \min (d-1, \log n-1)$ decisions since its invocation on $Ladder_n(G,f)$. Let $\pi_\delta$ be the current trail, then
	\begin{enumerate}
		\item The solver has yet to enter a conflict, and thus has not learned any clauses.
		\item The trail $\pi_\delta$ contains variables from at most $\delta$ different blocks $\ell^i$.
	\end{enumerate}
}

	To prove this lemma, we will first argue that  solver $\CNDRD$ is well-behaved on the formula
	\[ F : = \{ L^i \Rightarrow C^i : 0 \leq i \leq n-2\} \cup \{C^{n-1} \Rightarrow EQ\}. \]

\begin{lemma}\label{lemma: implications_never_conflict}
Let $\pi_{\delta-1}$ be the trail of $\CNDRD$ after $\delta-1 < \min(\log n-2,d-2)$ decisions from invocation on $F$. Suppose that the solver has yet to encounter a conflict. Let $\pi_\delta$ be obtained from $\pi_{\delta-1}$ by deciding some literal $x$ and applying unit propagation. Then,
	\begin{enumerate}
		\item The solver does not encounter a conflict from deciding $x$ and propagating.
		\item If $\pi_\delta$ does not contain an assignment $c_i = 0$, then $\pi_\delta$ contains assignments to at most $\delta$ $c$-variables.	
		\item The trail $\pi_\delta$ contains variables from at most $\delta$ different blocks $\ell^i$.
	\end{enumerate}
\end{lemma}

We will delay the proof, and first argue that until the solver has made at least $\min(\log n-1, d)$ decisions from its invocation, the clauses $C^i \Rightarrow Tseitin(G,f) $ will not impact its behaviour. That is, until the solver has made $\min(\log n-1, d)$ decisions, the solver will behave on $Ladder_n(G,f)$ as if it were run on $F$.

\begin{definition}
	We say that $\CNDRD$ \emph{behaves identically} on two formulas $F$ and $F'$ if given a trail $\pi$, a decision variable $x$, and an assignment $\alpha$, they produce the same trail $\pi'$ and learned clauses after setting $x=\alpha$.
\end{definition}

\begin{clm}\label{claim:identical_behaviour}
	Let $\pi_{\delta-1}$ be any trail of $\CNDRD$ after $\delta-1 < \min(\log n-2,d-2)$ decisions from invocation on $Ladder_n(G,f)$. Suppose that the solver has yet to encounter a conflict, and $\pi_{\delta-1}$ sets variables from at most $\delta-1$ blocks $\ell^i$. Then, for any $(x,\alpha)$, $\CNDRD$ behaves identically on $Ladder_n(G,f)$ and $F$.
\end{clm}
\begin{proof}
	Let $\pi_\delta$ be a trail (closed under unit propagations) which sets variables from less than $\min (\log n-2, d-2)$ blocks $\ell^i$ and let $x$ be the current decision variable and $\alpha$ be its assignment. Observe that the clauses of $C^i \Rightarrow Tseitin(G,f)$ each depend on variables from $d$ different blocks $\ell^i$. Thus, these clauses cannot cause any conflicts or propagations (and thus cause the behaviour on $F$ and $Ladder_n(G,f)$ to differ on decision $x=\alpha$) unless the clauses of $Ladder_n(G,f) \setminus \{C^i \Rightarrow Tseitin(G,f) : 0 \leq i \leq n-2\} =F$ cause the solver to set more than a single literal from a previously untouched block $\ell^i$. Because this propagation would depend only on the clauses of $F$, this would contradict Lemma~\ref{lemma: implications_never_conflict}. \qed
\end{proof}

\begin{proof}[of Lemma~\ref{lemma: solver_cant_screw_us}]
	The proof follows by combining Lemma~\ref{lemma: implications_never_conflict} and Claim~\ref{claim:identical_behaviour} together with induction on the number $\delta$ of decisions made thus far. Note that the hypotheses of Lemma~\ref{lemma: implications_never_conflict} and Claim~\ref{claim:identical_behaviour} are satisfied for $\delta=0$. \qed
\end{proof} 

It remains to prove Lemma~\ref{lemma: implications_never_conflict}. The proof of the first point will rely on the following lemma which says that $F$ is satisfiable provided we have not set all $c$-variables to $1$. This will be crucial in the proof of statement 1 in Lemma~\ref{lemma: implications_never_conflict}. Indeed, the unit propagator cannot propagate a formula into a conflict if there is a satisfying assignment.  

	\begin{clm}\label{claim: satisfying_implications}
Let $\pi$ be a partial assignment to the variables of $F$ such that $\pi$ is closed under unit propagation, $\pi$ does not falsify any clause of $F$, and either $\pi$ sets some $c_i=0$ or there are at least two $c$-variables that have not been set by $\pi$. 
Then, for any variable $x$ not set by $\pi$,  $F[\pi, x=\alpha]$ is satisfiable for any $\alpha \in \{0,1\}$. 
\end{clm}

\begin{proof}
	First, observe that $F[\pi]$ is satisfiable by the assignment $\rho$ extending $\pi$:
		\begin{align*}
		\rho(\ell^i_j) &:= 
		\begin{cases}
 			\pi(\ell^i_j) &\mbox{if $\pi$ sets $\ell^i_j$} \\
 			\beta &\mbox{if there exists $k \in [\log n]$ such that $\pi(\ell^i_k) = \beta$.} \\
 			0 &\mbox{otherwise} 
 	\end{cases} \\
 	\rho(c_m) &:= \begin{cases} \pi(c_m) &\mbox{if $\pi$ sets $c_m$}\\ 0 &\mbox{otherwise}\end{cases} &
	\end{align*}
	To see that this is satisfying, first note that $C^{n-1}$ is falsified as $\rho$ sets some $c_i = 0$. Therefore, the clauses of $C^{n-1} \Rightarrow EQ$ are satisfied. To see that the clauses of $L^i \Rightarrow C^i$ are satisfied, first observe that if $\rho$ sets all variables in a block $\ell^i$ to the same value then this falsifies $L^i$ and satisfies $L^i \Rightarrow C^i$.
	
	\begin{clm} \label{clm:at_most_one_conflicting_pair}
		There can be at most a single $0 \leq i \leq n-2$ for which there is $j,k \in [\log n]$ such that $\pi[\ell^i_j] = 0$ and $\pi[\ell^i_k] = 1$. Furthermore, the existence of such a pair $\ell^i_k, \ell^i_j$ forces $\pi$ to set $\pi(c_m) = bin(i,m)$ for all $m \in [\log n]$.
	\end{clm}
	\begin{proof}
		This follows because the clausal expansion of $L^i \Rightarrow C^i$ contains $(\ell^i_j \vee \neg \ell^i_k \vee c_m^{bin(i,m)})$ for all $m \in [\log n]$. Because $\pi$ is closed under unit propagation, these clauses must have forced $\pi$ to set $c_m = bin(i,m)$ for all $m \in [\log n]$. Therefore, if there was a $h \neq i$ for which some $\pi[\ell^h_p] =0$ and $\pi[\ell^h_q] = 1$, then the clause $(\ell^h_p \vee \neg \ell^h_q \vee c_m^{bin(h,m)})$ for which $bin(h,m) \neq bin(i,m)$ would be falsified by $\pi$. \qed
	\end{proof}
	Therefore, there is at most one block $\ell^i$ that for which all variables are not set to the same value by $\rho$. In which case $\pi(c_m) = bin(i,m)$ for all $m \in [\log n]$ and $L^i \Rightarrow C^i$ is satisfied.

	Next, we construct a satisfying assignment for $F[\pi,x=\alpha]$.  If $x = c_{m^*}$ for some $m^* \in [\log n]$, then define 
	\begin{align*} \rho_c(\ell^i_j) &:= \rho(\ell^i_j) &\forall~ 0\leq i \in n-2, j \in [\log n] 
	\\
	\rho_c(c_m) &:= \begin{cases} \rho(c_m) &\mbox{if $m \neq m^*$} \\
	\alpha &\mbox{otherwise}
 	\end{cases}	&\forall~ m \in [\log n]
	\end{align*}
	By Claim~\ref{clm:at_most_one_conflicting_pair}, does not exist $0\leq i \leq n-2$, $j,k \in [\log n]$ such that   $\pi[\ell^i_j]=0$ and $\pi[\ell^i_k]=1$, as this would imply that all $c$-variables had already been set by $\pi$. Thus, $\rho$, and therefore $\rho_c$, must set all variables in each block $\ell^i$ to the same value. As well, because $\pi$ has at least two $c$-variables unset (one of which is $x$) $\rho_c$ must set at least one $c_i=0$ and therefore satisfy $C^{n-1} \Rightarrow EQ$.
	
	Next, suppose that $x =\ell^{i^*}_{j^*}$ for some $0 \leq i^* \leq n-2$, $j^* \in [\log n]$. If every variable in block $\ell^{i^*}$ either has not been set by $\pi$, or has its value set to $\alpha$, then $\rho$ is a satisfying assignment to $F[\pi,x=\alpha]$. So, suppose that this is not the case. Assume wlog that $\alpha = 0$ and let $k^* \in [\log n]$ be such that $\pi[\ell^{i^*}_{k^*}] = 1$. Then, we claim that there does not exist $i \neq i^*$ such that there is $p,q \in [\log n]$ for which $\pi[\ell^i_p] = 0$ and $\pi[\ell^i_q]=0$. Indeed, by Claim~\ref{clm:at_most_one_conflicting_pair} this would imply that the $\pi(c_i) = bin(i,m)$ for all $i \in [\log n]$. Let $m^*$ be such that $bin(i,m) \neq bin(i,m^*)$ then, because $\pi$ sets $c_{m^*}$ and $\ell^{i^*}_{k^*}$ the clause of $L^{i^*} \Rightarrow C^{i^*}$,
	 \[(\ell^{i^*}_{j^*} \vee \neg \ell^{i^*}_{k^*} \vee c_{m^*}^{bin(i,m^*)})[\pi] = (\ell^{i^*}_{j^*})\] 
	 That is, it is unitary under $\pi$, contradicting our assumption that $\pi$ is closed under unit propagation. Next, observe that for every $c_m$ for $m \in [\log n]$ that $\pi$ has set, we must have $\pi[c_m] = bin(i^*,m)$, for every $p \in [\log n]$, otherwise the clause of $L^{i^*} \Rightarrow C^{i^*}$, 
	\[ (\ell^{i^*}_p \vee \neg \ell^{i^*}_k \vee c_m^{bin(i^*,m)})[\pi] = (\ell^{i^*}_{j^*}),\]
	which contradicts our assumption that $\pi$ is closed under unit propagation. 	
	Define $\rho_\ell$ extending $\pi$ as 
	\begin{align*}
		&\rho_\ell(\ell_k^j) := \rho(\ell^j_k) &\forall j \neq i^*, 0 \leq j \leq n-2,  k \in [\log n] 
		\\
		&\rho_\ell(\ell_k^{i^*}) := \begin{cases} \pi(\ell^{i^*}_k) &\mbox{ if $\pi$ sets $\ell^{i^*}_k$} \\
		\alpha &\mbox{ if $\ell^{i^*}_k = x$}\\ 0&\mbox{otherwise} \end{cases} &\forall k \in [\log n] 
		\\
		&\rho_\ell(c_m) := bin(i^*,m) &\forall m \in [\log n ]
	\end{align*}
	We claim that $\rho_\ell$ is a satisfying assignment. Indeed, as argued previously all clauses $L^i \Rightarrow C^i$ for $i \neq i^*$ are satisfied because $\rho_\ell$ sets all variables in the block $\ell^i$ to the same value. $L^{i^*} \Rightarrow C^{i*}$ is satisfied because we have set $c_m = bin(i^*,m)$ for all $m \in [\log n]$. Finally, $C^{n-1}$ is falsified because $i^* \neq n-1$ and therefore $C^{n-1} \Rightarrow EQ$ is satisfied. \qed
\end{proof}

Finally, we are ready to prove Lemma~\ref{lemma: implications_never_conflict}.

\begin{proof}[of Lemma~\ref{lemma: implications_never_conflict}]
	The proof is by induction on $\delta$. That it holds for $\delta=0$ is trivial because $F$ does not contain any unit clauses. 
	
	To prove (1), observe that the trail $\pi_{\delta-1}$ satisfies the hypothesis of Claim~\ref{claim: satisfying_implications}. Therefore, $F[\pi,x=\alpha]$ is satisfiable for every $\alpha \in \{0,1\}$. Because thee unit propagator is sound, the unit propagator will not falsify a clause of the formula and cause the solver to enter a conflict. 
	
	To prove (2), suppose that $\pi_{\delta-1}$ sets at most $\delta-1 < \min( \log n-2, d-2)$ many $c$-variables, all of which are set to $1$. We argue that if $\pi_\delta$ does not set some $c_i=0$, then it sets at most one additional $c$-variable to $1$. First, observe that because each clause of $C^{n-1} \Rightarrow EQ$ depends on at least $\log n$ $c$-variables, these clauses cannot cause any propagations unless setting $x=\alpha$ causes the clauses of $\{C^i \Rightarrow L^i : i \in [n-2]\}$ to propagate more than one $c$-variable in the $\delta$th during this decision level (i.e. $\pi_{\delta}$ sets at least $2$ more $c$-variables than $\pi_{\delta-1}$). Therefore, we will restrict attention to the clauses of $\{C^i \Rightarrow L^i : i \in [n-2]\}$ and show that these clauses cannot cause the solver to set more than a single $c$-variable in the $\delta$th decision level without setting some $c_i$ to $0$. 
	
	First, because $\pi_\delta$ does not set any $c_i =0$, there cannot exist $0 \leq i \leq n-2$, $j,k \in [\log n]$ such that $\pi_\delta[\ell^i_j]=0$ and $\pi_\delta[\ell^i_k] =1$. This is because there exists $m \in [\log n]$ such that $bin(i,m)=0$ as $i < n-1$, and so the clause $(\ell^i_j \vee \neg \ell^i_k \vee c_m^{bin(i,m)})$ of $L^i \Rightarrow C^i$ would propagate some $c_m = bin(i,m) = 0$. Therefore, if $x$ is a variable $\ell^i_j$ then all variables in $\ell^i$ that have been set by $\pi_{\delta-1}$ must take the same value $\alpha \in \{0,1\}$. If no variable in $\ell^i$ has previously been set, then $x$ may be set freely. In the second case, $x$ must be set to $\alpha$ so as to not propagate some $c_i$ to $0$. This cannot cause a propagation either, as if we can only set $\ell^i$ variables to the same value, then each clause $(\ell^i_j \vee \neg \ell^i_k \vee c_m^{bin(i,m)})$ of $L^i \Rightarrow C^i$ must either contain an unset $\ell^i$-variable, or be satisfied. In both of these cases, we have set not set any additional $c$-variables.
	
	Finally, suppose that $x = c_{m^*}$ for some $m^* \in [\log n]$. We claim that setting $x=1$ cannot propagate any $c$-variables without setting some $c_i = 0$. The only way that setting a $c$-variable can propagate another $c$-variable is if there exists $0 \leq i \leq n-2, k,h \in [\log n]$ such that $\pi_{\delta-1}[\ell^i_k] = 1$ and $\pi[\ell^i_h ]= 0$ for which $bin(i,m^*) = 0$. In this case, consider the pair of clauses from $L^i \Rightarrow C^i$,
\[ (\ell^i_j \vee \neg \ell^i_k \vee c_{m^*}^{bin(i,{m^*})}),~ (\ell^i_{h} \vee \neg \ell^i_{j} \vee c_{m}^{bin(i,m)}). \]
 Then, setting $x= c_{m^*} = 1$ will propagate $c_m = bin(i,m)$. 
However, in the clausal expansion of $L^i \Rightarrow C^i$ there exists a clause $(\ell^i_h \vee \neg \ell^i_j \vee c_{p}^{bin(i,p)})$ for every $p \in [\log n]$, and so this would propagate \emph{every} $c$-variable. Because $0 \leq i \leq n-2$, there exists some $q \in [\log n]$ such that $bin(i,q) =0$. Thus, this would set some $c_q = 0$, contradicting our assumption that $\pi_{\delta}$ does not set any $c$-variable to $0$.

To prove (3), assume by induction that $\pi_{\delta-1}$ has set variables from at most $\delta -1$ different blocks $\ell^i$. Let $i^*$ be such that $\pi_{\delta-1}$ does not set any variables in block $\ell^{i^*}$. The only clauses that involve variables from $\ell^{i^*}$ belong to  $L^{i^*} \Rightarrow C^{i^*}$ and $C^{n-1} \Rightarrow EQ$. The clauses of the latter each depend on all $\log n$ $c$-variables. Either some $c$-variable has been set to $0$ by $\pi_\delta$ or by (2), at most $\delta < \log n-1$ $c$-variables have been set by $\pi_{\delta}$. In the first case, $C^{n-1}$ is falsified and thus all clauses of $C^{n-1} \Rightarrow EQ$ are satisfied. In the second case, $\pi_{\delta}$ has not set enough $c$-variables for these clauses to cause any propagations. 

Next, consider the clauses of $L^{i^*} \Rightarrow C^{i^*}$ which are all of the form $(\ell^{i^*}_j \vee \neg \ell^{i^*}_k \vee c_m^{bin(i^*,m)})$ for some $j,k,m \in [\log n]$. By assumption, both $\ell^{i^*}_j$ and $\ell^{i^*}_k$ have not been set by $\pi_\delta$, and therefore cannot be unit propagated unless some variable from the block $\ell^{i^*}$ is set. Therefore, variables from block $\ell^{i^*}$ will be set by $\pi_\delta$ only if $x$ belongs to $\ell^{i^*}$. As $x$ can only belong to at most one block $\ell^i$, variables from at most one additional block $\ell^i$ can be set by $\pi_\delta$ than were set by $\pi_{\delta-1}$. \qed
\end{proof}

\subsection*{2.\ \ \ \  Proof of  Lemma~\ref{lemma:reduction_to_tseitin} (See page \pageref{lemma:reduction_to_tseitin})}

\textbf{Lemma~\ref{lemma:reduction_to_tseitin}. }
{\it
	Let $\pi$ be any restriction such each clause of $Ladder_n(G,f)[\pi]$ is either satisfied or contains at least two unassigned variables, and $\pi$ implies Tseitin.
	Suppose that $\pi$ sets variables from at most $\delta$ blocks $\ell^i$. Then there exists a restriction $\rho_\pi^*$  that sets at most $\delta$ variables of $\textit{Tseitin}(G,f)$ such that 
	\[ Res(Ladder_n(G,f)[\pi] \vdash \emptyset ) \geq Res(Tseitin(G,f)[\rho_\pi^*] \vdash \emptyset ).\]
}

The proof of this lemma will rely on the fact that the size of Resolution proofs is closed under both restrictions and projections. 

\begin{definition}
Let $F$ be any formula over the variables $\{x_1,\ldots, x_n\}$ define a \emph{projection} as any map $I: \{x_1,\ldots, x_n\} \rightarrow \{x_1,\ldots,x_k\}$ for $k \leq n$. Define the \emph{projected formula} $proj_I(F)$ under the projection $I$ by replacing each occurrence of a variable $x_j$ in $F$ by $I(x_j)$.
\end{definition}
That is, a projection may identify several variables as a single variable.  

\begin{lemma}
\label{lemma: Projections_Preserve_Size}
	Let $F$ be any CNF formula. Then,
	\begin{itemize}
		\item For any restriction $\rho \in \{0,1,*\}^n$, $Res(F \vdash \emptyset) \geq Res(F[\rho] \vdash \emptyset)$.
		\item For any projection $I$, $Res(F \vdash \emptyset) \geq Res(proj_I(F) \vdash \emptyset)$. 
	\end{itemize}
\end{lemma}
We omit the proof, but remark that it is folklore. 


\begin{proof}[of Lemma~\ref{lemma:reduction_to_tseitin}]
	Let $\pi$ be any restriction that implies Tseitin and such that each clause of $Ladder_n(G,f)[\pi]$ is either satisfied or contains at least two unset literals (one should think of $\pi$ as a trail produced by $\CNDRD$). Furthermore, suppose that $\pi$ sets variables from at most $\delta$ blocks $\ell^i$. To prove the lemma we will construct a restriction and projection pair $\rho_\pi$, $I_\pi$ consistent with $\pi$ such that
	\begin{align}
 		proj_{I_\pi} \big(Ladder_n(G,f)[\rho_\pi] \big) = Tseitin(G,f)[\rho_\pi^*], \label{eq:equality_of_projected_instances}
 	\end{align}
	for an associated restriction $\rho_\pi^*$ which sets at most $\delta$ variables. The Lemma~\ref{lemma:reduction_to_tseitin} will follow by applying Lemma~\ref{lemma: Projections_Preserve_Size}. 
	
	Define $\rho_{\pi}$ to be the restriction consistent with $\pi$ which sets all remaining $c$-variables, and sets the remaining variables in any block $\ell^i$ that has been partially set by $\pi$ to the same value.  Formally,
	\begin{align*} 
	\rho_\pi(\ell^i_j) &:= \begin{cases} * &\mbox{if $\pi$ does not set any variable in block $\ell^i$} \\
	0 &\mbox{if there exists $p,q \in [\log n]$ such that $\pi[\ell^i_p] \neq \pi[\ell^i_q] \neq *$ } \\
	\beta &\mbox{if $\pi[\ell^i_k] \in \{\beta,*\}$ for every $k \in [\log n]$, for some $\beta \in \{0,1\}$} 
 	\end{cases}	\\
 	\rho_\pi(c_m) &:= \begin{cases} \pi(c_m) &\mbox{if $\pi$ sets $c_m$} \\ 0 &\mbox{otherwise} \end{cases} 
 \end{align*}
Define the restriction $\rho_\pi^*$ as the projection of $\rho_\pi$ to the variables $\{\ell^1_1,\ldots, \ell^{n-1}_1 \}$ of $\textit{Tseitin}(G,f)$, 
\[ \rho_\pi^*(\ell^i_1):=\begin{cases} * &\mbox{if $\pi$ does not assign any variable in block $\ell^i$} \\
 \rho_\pi(\ell^i_1) &\mbox{otherwise.} 
 \end{cases}
 \]
 Observe that $|(\rho_\pi^*)^{-1}(1)| + |(\rho^*_\pi)^{-1}(0)| \leq \delta$ because $\pi$ sets variables from at most $\delta$ blocks $\ell^i$. Finally, define the projection $I_\pi$ as 
 \begin{align*} 
I_\pi(\ell^i_j)&:=\ell^i_1,
\end{align*}
which contracts each variable in the block $\ell^i$ of $Ladder_n(G,f)[\rho_\pi]$ to the single variable $\ell^i_1$ from block $\ell^i$ that $Tseitin(G,f)[\rho_\pi^*]$ depends on. 

Next, we argue that none of the clauses $\{L^i \Rightarrow C^i : 0 \leq i \leq n-2\} \cup \{C^{n-1} \Rightarrow EQ$ in  $Ladder_n(G,f)[\rho_\pi]$ are falsified (it is okay if a clause $C^i \Rightarrow Tseitin(G,f) $ is falsified as the same clause will be falsified in $Tseitin(G,f)[\rho_\pi^*]$ by construction and so the conclusion of Lemma~\ref{lemma:reduction_to_tseitin} holds vacuously). Each clause of $L^i \Rightarrow C^i$ is satisfied by any assignment that sets $\ell^i_j$ to the same value $\beta$. Therefore, we need only consider the case when $\pi[\ell^i_p] = 0$ and $\pi[\ell^i_q] =1$ for some $p,q \in [\log n]$. We claim that there can be at most one such $i$ for which this occurs. The constraint $L^i \implies C^i$ consists of clauses $(\ell^i_p \vee \neg \ell^i_q \vee c_m^{bin(i,m)})$ for every $m \in [\log n]$. Thus, $\pi(c_m) = bin(i,m)$ for every $m \in [\log n]$ as otherwise this would contradict our assumption that all clauses are either satisfied or contain at least 2 unassigned variables. Therefore, the existence of some $j \neq i$ and $p',q' \in [\log n]$ such that $\pi[\ell^j_{p'}] =0$ and $\pi[\ell^j_{q'}] =1$ would mean that the clauses $(\ell^j_{p'} \vee \neg \ell^j_{q'} \vee c_m^{bin(j,m)})$ for $m \in [\log n]$ would contradict our assumption, as there must exist some $m' \in [\log n]$ such that $bin(j,m) \neq bin(i,m)$. Thus, $\rho_\pi$ does not falsify any $L^i \Rightarrow C^i$. Finally, observe that $C^{n-1} \Rightarrow EQ$ is falsified only by assignments that set $c_i = 1$ for all $i \in [\log n]$. Indeed, the clauses of this constraint consist of $( \neg c_1 \vee \ldots \vee \neg c_{\log n} \vee \ell^i_j \vee \neg \ell^p_q)$ for every $0 \leq i,p \leq n-2$, $j,q \in [\log n]$. Under our assumption that there is some $\ell^i_j, \ell^p_q$ such that $\pi[\ell^i_p] = 0$ and $\pi[\ell^i_q] =1$ and that $\pi$ does not falsify any clause, $\pi$ cold not have set $c_i = 0$ for all $i \in [\log n]$. Thus, $C^{n-1} \Rightarrow EQ$ is satisfied under $\rho_\pi$.

It remains to argue that (\ref{eq:equality_of_projected_instances}) holds. All clauses of  $L^i \Rightarrow C^i$ that are not satisfied by $\rho_{\pi_\delta}$ have all of their $\ell$-variables unset. These clauses are of the form $(\ell^i_j \vee \neg \ell^i_k \vee c_m^{bin(i,m)})$. Under the projection $I_{\pi_\delta}$ this clause becomes $(\ell^i_1 \vee \neg \ell^i_1 \vee c_m^{bin(i,m)})$, which is the identically $1$ clause (it contains both $\ell^i_1$ and $\neg \ell^i_1$). 
\qed
\end{proof}

\end{document}